%% file: secure-cbf-tcc-2025.tex
\documentclass{iacrcc}
\usepackage{amsthm,amsmath}
\usepackage[T1]{fontenc}
\usepackage[advantage,adversary,asymptotics,complexity,keys,lambda,landau,logic,sets,mm,oracles,primitives,probability,notions,sets,operators]{cryptocode} 
\usepackage{color, colortbl, tcolorbox}
\usepackage{url, array, cite, multirow, mathtools, framed, newfloat,microtype,algpseudocode,algorithm}
\usepackage[T1]{fontenc}
\usepackage{diagbox, booktabs, appendix, alltt}
\usepackage{listings, xcolor, afterpage, bm, xspace}
\usepackage[inkscapelatex=false]{svg}
\usepackage[shortlabels]{enumitem}
\usepackage[multiple]{footmisc}
\usepackage[cal=boondox]{mathalfa}
\usepackage{stmaryrd}
\usepackage{caption}
\usepackage{tcolorbox}
\usepackage[english]{babel}
\usepackage{comment}
\usepackage{subcaption}
\usepackage[mathscr]{euscript}
\usepackage{tabularx}
\usepackage{tikz}
\usetikzlibrary{arrows.meta, positioning, fit}
\usepackage{centernot}

\usepackage{soul}

\usepackage{graphicx}

\AtBeginDocument{
    \usepackage[capitalise]{cleveref}
}

\input{macros}

\begin{document}
\title[running = {Adversarially Robust Bloom Filters}
]{Adversarially Robust Bloom Filters: Privacy, Reductions, and Open Problems}
%
%
\addauthor[inst=1,
           email={hayder.research@gmail.com}]{Hayder Tirmazi}
%
%
%
\addaffiliation[department={CUNY},
                city={New York}, 
                country={USA}]{City College of New York}

\maketitle              
\begin{abstract}
A Bloom filter is a space-efficient probabilistic data structure that represents a set $S$ of elements from a larger universe $U$. This efficiency comes with a trade-off, namely, it allows for a small chance of false positives. When you query the Bloom filter about an element x, the filter will respond 'Yes' if $x \in S$. If $x \notin S$, it may still respond 'Yes' with probability at most $\varepsilon$. We investigate the adversarial robustness and privacy of Bloom filters, addressing open problems across three prominent frameworks: the game-based model of Naor-Oved-Yogev (NOY), the simulator-based model of \filic{} et. al., and learning-augmented variants. We prove the first formal connection between the \filic{} and NOY models, showing that \filic{} correctness implies AB-test resilience. We resolve a longstanding open question by proving that PRF-backed Bloom filters fail the NOY model's stronger BP-test. Finally, we introduce the first private Bloom filters with differential privacy guarantees, including constructions applicable to learned Bloom filters. Our taxonomy organizes the space of robustness and privacy guarantees, clarifying relationships between models and constructions.



\end{abstract}

\begin{textabstract}
TODO
\end{textabstract}

\input{sections/introduction}
\input{sections/preliminaries}
\input{sections/problems}
\input{sections/privacy}
\input{sections/bridging-models}

\input{sections/prf-standard-bloom-filter}
\input{sections/conclusion}
\input{sections/acknowledgements}

\bibliography{paper}

\end{document}

%% file: sections/introduction.tex
\section{Introduction}

A Bloom filter is a probabilistic data structure that encodes a set $S$ from some large but finite universe $U$. Bloom filters are used to answer membership queries, i.e., for some $x \in U$, is $x \in S$? Bloom filters use less memory than explicitly encoding $S$, but at the cost of false positives. For any $x$ if $x \in S$, the Bloom filter will return true with probability $1$. If $x \notin S$, the Bloom filter might still incorrectly return true with probability at most $\varepsilon$ for some $\varepsilon \in [0, 1]$. Bloom filters are widely deployed in critical real-world systems such as Google's LevelDB~\cite{leveldbbloom}, Meta's RocksDB~\cite{rocksdbbloom}, and the Linux Kernel~\cite{kernelbpfmap}. This has made the adversarial robustness of Bloom filters a growing concern~\cite{gerbetsecurity,moni1}. 

Bloom filters have historically only been analyzed in a non-adversarial setting where the false positive probability of an element $x$ uniformly randomly chosen from $U$ is computed over the internal randomness of the Bloom filter construction~\cite{MitzenmacherBroder2004}. A series of recent works has focused, instead, on the performance of Bloom filters in the presence of adversaries. Naor et al.~\cite{moni1,naor2022bet,lotan2025adversarially} propose game-based robustness notions such as the Always-Bet (AB) and Bet-or-Pass (BP) tests for Bloom filters. Almashaqbeh et al.~\cite{almashaqbeh2024adversary} extend these game-based notions to learned Bloom filters, which are a variant of Bloom filters that use machine learning models. \filic{} et al. develop simulator-based robustness notions. Despite significant recent progress, the relationships between these models remain unclear, and several important problems remain unanswered.

The goal of this work is to unify and advance this recently developed theory of Bloom filter adversarial robustness. We articulate $10$ open problems that span adversarial models, construction styles, and privacy goals. We contribute to $3$ of these problems and leave the remaining $7$ problems as open directions for research in this area. In terms of our discussed open problems (Section~\ref{sec:openproblems}), we partially solve Problem~\ref{prob:notion-unification}, and mostly solve Problems~\ref{prob:prf-standard-bloom} and~\ref{prob:bloom-privacy}.
\medskip

\noindent\textbf{Private Bloom filters}: We introduce the first Bloom filter constructions that satisfy differential privacy guarantees. In particular, we introduce two constructions, the Mangat filter and the Warner filter, based on well-known randomized response mechanisms. We also fill a gap in this area by showing that Mangat's randomized response satisfies the notion of asymmetric differential privacy.\medskip

\noindent\textbf{Bridging NOY model and \filic{} model}: We show that \filic{} correctness implies AB-test resilience, marking the first known formal connection between these security frameworks. We also show that AB-test and BP-test resilience do not imply \filic{} correctness.\medskip

\noindent\textbf{PRF-backed Standard Bloom filter}: We prove that, in all practical cases, a PRF-backed Standard Bloom filter does not satisfy the NOY model's notion of BP-test resilience. This was left as an open question by Naor and Oved~\cite{naor2022bet}.\medskip

After covering related work, the remainder of this paper is organized as follows. We discuss preliminaries in Section~\ref{sec:preliminaries}. Section~\ref{sec:openproblems} is dedicated to discussing the open problems we enumerate in this work and providing a taxonomy of them. Section~\ref{sec:privatebloomfilters} covers the private Bloom filter constructions we introduce in this work. Section~\ref{sec:bridgingnoyfilic} introduces formal connections between the NOY model and the \filic{} model. Section~\ref{sec:prfbackedsbfbptest} proves the result regarding PRF-backed Standard Bloom filters not being resilient under the BP-test.

\subsection{Related Work}

Gerbet et al.~\cite{gerbetsecurity} suggests practical attacks on Bloom filters and the use of universal hash functions and MACs to mitigate a subset of those attacks. Naor and Yogev~\cite{moni1} define an adversarial model for Bloom filters and provide a method for constructing adversary-resilient Bloom filters. Naor and Oved~\cite{naor2022bet} present several robustness notions in a generalized adversarial model for Bloom filters. Clayton et al.~\cite{claytonetal} and \filic{} et al.~\cite{filic1} provide secure constructions for Bloom filters using a game-based and a simulator-based model, respectively. Reviriego et al.~\cite{reviriego1} propose a practical attack on learned Bloom filters. They suggest possible mitigations, e.g., swapping to a classical Bloom filter upon attack detection. Almashaqbeh et al.~\cite{almashaqbeh2024adversary} propose provably secure learned Bloom filter constructions by extending the adversarial model of Naor et al.~\cite{moni1,naor2022bet}.

Many works, including Sengupta et al.~\cite{sengupta2017sampling}, Reviriego et al.~\cite{reviriego2022privacy}, and Galan et al.~\cite{galan2022privacy} have shown that Bloom filters are vulnerable to set reconstruction attacks, i.e., given the internal state of a Bloom filter it is possible to infer the set the Bloom filter stores with high probability. Bianchi et al.~\cite{bianchi2012better} provide privacy metrics for Bloom filters. Bianchi et al.'s metrics are based on k-anonymity~\cite{sweeney2002k}. \filic{} et al.~\cite{filic1} propose a simulator-based notion of privacy for Bloom filters based on information leakage profiles. 
They provide privacy bounds for Bloom filters that use pseudo-random functions on their input set. \filic{} et al.'s proposal does not achieve meaningful privacy for Bloom filters whose input sets have low min-entropy, and their notion of Elem-Rep privacy is not immune to set reconstruction attacks from computationally unbounded adversaries. Concurrently and independently of our work, Ke et al.~\cite{ke2025dpbloomfilter} propose a differentially private Bloom filter construction in a preprint dated February 2, 2025. Our Warner filter, introduced in the first version of this preprint publicly available on January 27, 2025, provides a similar differential privacy guarantee using Warner’s randomized response mechanism. To the best of our knowledge, ours is the earliest work to formally analyze differential privacy for Bloom filters in this setting. We are also not aware of any prior work that analyzes the privacy of learned Bloom filters.

%% file: sections/preliminaries.tex
\section{Preliminaries}\label{sec:preliminaries}

For a set $S$, $x \sample S$ denotes that $x \in S$ is sampled uniformly at random from $S$. For $n \in \mathbb{N}$, $[n]$ denotes the set $\{ 1, \cdots, n\}$. A Standard Bloom filter~\cite{bloom1970} is a bit string $M = \bin^{m}$ of length $m$ bits indexed over $[m]$, along with $k$ different hash functions $h_i$. Each $h_{i}$ maps an element from $x \in U$ to an index value within $M$, i.e, $h_{i}: U \mapsto [m]$. Let $\standardbloom$ be a Standard Bloom filter. To encode a set $S$ in $\standardbloom$, initialize a bit string $M$ with all bits set to $0$. Then, take each element $x \in S$, and for $i \in [k]$, set the bit corresponding to index $h_{i}(x)$ of $M$ to $1$. To answer a query for some element $x \in U$ in $\standardbloom$, return $1$ if every bit in $M$ corresponding to indices $h_{i}(x)$ is $1$. Otherwise, return $0$.

\subsection{Naor-Oved-Yogev Model}

The first major adversarial model for Bloom filters is developed in a series of papers by Naor et al.~\cite{moni1, naor2022bet, lotan2025adversarially}. We will refer to this as the Naor-Oved-Yogev (NOY) Model in this paper. For a finite universe $U$ of cardinality $u$, consider a set $S \subseteq U$. Naor and Yogev~\cite{moni1} define a Bloom filter as a data structure composed of two algorithms: a construction algorithm and a query algorithm.

\begin{definition}
    Let $\bloom = (\bloomconstruct, \bloomquery)$ be a pair of $\ppt$ algorithms. $\bloomconstruct$ takes a set $S \subseteq U$ and returns a representation $M$. $\bloomquery$ takes a representation $M$ and a query element $x \in U$ and outputs a value in $\bin$. $\bloom$ is an $(n, \varepsilon)$-Bloom filter if for all sets $S \subseteq U$ of cardinality $n$, the following two properties hold.

    \begin{enumerate}
        \item Completeness: For any $x \in S$: $\Pr[\bloomquery(\bloomconstruct, x) = 1] = 1$
        \item Soundness: For any $x \notin S$: $\Pr[\bloomquery(\bloomconstruct(S), x) = 1] \leq \varepsilon$
    \end{enumerate}
    where the probabilities are taken over the random coins of $\bloomconstruct$ and $\bloomquery$~\cite{moni1}.
\end{definition}
We assume $\bloom$ always has this format in this paper. If a Bloom filter's query algorithm cannot change the set representation, $M$, it is called a \textit{steady} Bloom filter. Otherwise, it is called an \textit{unsteady} Bloom filter. For simplicity, we assume a steady Bloom filter in our results (similar to prior work~\cite{naor2022bet,lotan2025adversarially}), unless explicitly stated otherwise. 

Naor and Oved~\cite{naor2022bet} define $\texttt{AdaptiveGame}_{\adv,t}(\lambda)$, a unified security game for Bloom filters. The game has an adversary $\adv = (\advconstruct, \advquery)$. $\advconstruct$ chooses any set $S \subseteq U$. $\advquery$ takes set $S$ and performs adaptive queries to a Bloom filter $\bloom$. $\advquery$ is also allowed oracle access to the query algorithm $\bloomquery$. $\lambda$ is the security parameter, it is given to $\advconstruct$ and $\bloomconstruct$. $t$ denotes the number of queries $\advquery$ is allowed to perform.

\begin{definition}

\(\texttt{AdaptiveGame}_{\adv, t}(\lambda)\) ~\cite{naor2022bet}

\begin{enumerate}
    \item Adversary $\advconstruct$ takes $1^{\lambda + n\log{u}}$ and returns a set $S \subseteq U$ of cardinality $n$.
    \item $\bloomconstruct$ takes $(1^{\lambda + n\log{u}}, S)$ and builds representation $M$.
    \item Adversary $\advquery$ takes $(1^{\lambda + n\log{u}}, S)$ and oracle access to $\bloomquery(M, \cdot)$, and performs at most $t$ adaptive queries $x_{1}, \cdots, x_{t}$ to $\bloomquery(M, \cdot)$.
\end{enumerate}
\end{definition}

There are many security notions based on \texttt{AdaptiveGame}. The two relevant to our work are the Always-Bet (AB) Test, proposed by Naor and Yogev~\cite{moni1}, and a stronger notion called Bet-or-Pass (BP) Test, proposed by Naor and Oved~\cite{naor2022bet}.

\subsubsection{Always-Bet (AB) Test}

In the Always-Bet (AB) test, adversary $\adv$ plays \(\texttt{AdaptiveGame}_{\adv, t}(\lambda)\) and is then required to return an $x^{*} \in U$. $\adv$ wins if $x^{*}$ is an unseen false positive.\medskip

\noindent \textbf{AB Test} \(\texttt{ABTest}_{\adv,t}(\lambda)\)~\cite{moni1,naor2022bet}:
\begin{enumerate}
    \item $\adv$ plays \(\texttt{AdaptiveGame}_{\adv, t}(\lambda)\). $S$ is the set $\adv$ chose in the game and $\{ x_{1}, \cdots, x_{t} \}$ are the queries $\adv$ performed in the game.
    \item $\adv$ returns $x^{*} \notin S \cup \{ x_{1}, \cdots, x_{t}\}$.
    \item If $\bloomquery(M, x^{*}) = 1$, return 1. Return 0, otherwise.
\end{enumerate}

\begin{definition}
    An $(n, \varepsilon)$-Bloom filter $\bloom$ is $(n, t, \varepsilon)$-AB test resilient if for any adversary $\adv$, there exists a negligible function $\negl[]$ such that
    \[
        \Pr[\texttt{ABTest}_{\adv, t}(\lambda) = 1] \leq \varepsilon + \negl[\lambda]
    \]
    \noindent where the probabilities are taken over the internal randomness of $\bloom$ and $\adv$.~\cite{moni1,naor2022bet}
\end{definition}

Naor and Yogev~\cite{moni1} introduce a Bloom filter construction that is robust under the AB-test. Their construction is based on keyed pseudo-random permutations. Similar to other papers~\cite{almashaqbeh2024adversary}, we will refer to this construction as the Naor-Yogev (NY) filter.

\subsubsection{Bet-or-Pass (BP) Test}

In the Bet-or-Pass (BP) test, adversary $\adv$ can \textit{pass} instead of returning an unseen false positive $x^{*}$. $\adv$ plays \(\texttt{AdaptiveGame}_{\adv, t}(\lambda)\) and is then required to return $(b, x^{*})$. $b \in \bin$ represents whether $\adv$ wants to \textit{bet} on the returned element $x^{*}$, or \textit{pass}. $\adv$'s win is based on a profit $C_{\adv}$ defined by the test.\medskip

\noindent\textbf{BP Test} \(\texttt{BPTest}_{\adv, t}(\lambda)\)~\cite{naor2022bet}:

\begin{enumerate}
    \item $\adv$ plays \(\texttt{AdaptiveGame}_{\adv, t}(\lambda)\). $S$ is the set $\adv$ chose in the game and $\{ x_{1}, \cdots, x_{t} \}$ are the queries $\adv$ performed in the game.
    \item $\adv$ returns $(b, x^{*})$ where $x^{*} \notin S \cup \{ x_{1}, \cdots, x_{t}\}$.
    \item Return $\adv$'s profit $C_{\adv}$, defined as \begin{equation*}
    C_{\adv} = \begin{cases} 
        \frac{1}{\varepsilon}, &\text{if } x^{*} \text{is a false positive and } b = 1,\\
        -\frac{1}{1 - \varepsilon}, & \text{if } x^{*} \text{is not a false positive and } b = 1,\\
        0, & \text{if } b = 0.
        \end{cases} \end{equation*}
\end{enumerate}

\begin{definition}
    An $(n, \varepsilon)$-Bloom filter $\bloom$ is $(n, t, \varepsilon)$-BP test resilient if for any adversary $\adv$, there exists a negligible function $\negl[]$ such that
    \[
        \expect{C_{\adv}} \leq \negl[\lambda]
    \]
    \noindent where the probabilities are taken over the internal randomness of $\bloom$ and $\adv$.~\cite{naor2022bet}
\end{definition}

The Bet-or-Pass test is the strongest security notion~\cite{lotan2025adversarially} currently defined in the \texttt{AdaptiveGame} setting. Naor and Oved~\cite{naor2022bet} prove that BP test is strictly stronger than AB test. Specifically, they prove that $(n, t, \varepsilon)$-BP test resilience implies $(n, t, \varepsilon)$-AB test resilience, and the converse implication is false. 

Naor and Oved~\cite{naor2022bet} introduce a Bloom filter construction that is robust under the BP-test. Their construction builds on an earlier Cuckoo hashing-based construction by Naor and Yogev~\cite{moni1} and relies on keyed pseudo-random functions. Similar to other papers~\cite{almashaqbeh2024adversary}, we will refer to this construction as the Naor-Oved-Yogev (NOY) Cuckoo filter.

\subsubsection{Universe and Adversary Types}

A universe $U$ is \textit{small} if its cardinality $u \in \bigO{\poly[t,n,\lambda]}$, otherwise $U$ is \textit{large}. An adversary with a query budget $t$ can query at most a negligible fraction of the elements of a large universe. Adversary $\adv$ in \(\texttt{AdaptiveGame}_{\adv, t}(\lambda)\) can either be computationally bounded, i.e., running in probabilistic polynomial time ($\ppt$), or computationally unbounded, i.e., not restricted to $\ppt$ but still bounded by the number of queries $t$. Consider $(n, t, \varepsilon)$-resilient Bloom filter $\bloom$ under the AB test or BP test. If $\bloom$ is $(n, t, \varepsilon)$-resilient for any polynomial number of queries $t \in \bigO{\poly[n, \lambda]}$ under a computationally bounded adversary, $\bloom$ is called $(n, \varepsilon)$-\textit{strongly-resilient}~\cite{moni1}. If $\bloom$ is resilient for at most $t$ queries, under a computationally unbounded adversary, then $\bloom$ is called \textit{t-resilient}~\cite{lotan2025adversarially}.

\subsection{\filic{} Model}

The second major adversarial model for Bloom filters was introduced by \filic{}, Paterson, Unnikrishnan, and Virdia~\cite{filic1,virdia2024note,filic2025deletions}. Our presentation of the model modifies \filic{} et al.'s original notation for easier comparison with the model of Naor et al. The \filic{} model allows inserting elements into a Bloom filter $\bloom$ after $\bloom$'s construction. We can define this in Naor's notation by adding a third polynomial time algorithm, $\bloominsert$, that does insertions.

\begin{definition}
    Let $\bloom = (\bloomconstruct, \bloomquery, \bloominsert)$ be a 3-tuple of $\ppt$ algorithms. $\bloomconstruct$ takes a set $S \subseteq U$ and returns a representation $M$. $\bloomquery$ takes a representation $M$ and a query element $x \in U$ and outputs a value in $\bin$. $\bloominsert$ takes a representation $M$ and a query element $x \in U$ and outputs a new representation $M^{\prime}$ encoding the set $S \cup \{ x \}$. $\bloom$ is an $(n, \varepsilon)$-insertable Bloom filter if for all sets $S \subseteq U$ of cardinality $n$ and for at most $\ell$ insertions, the following four properties hold.

    \begin{enumerate}
        \item Completeness: For any $x \in S$: $\Pr[\bloomquery(\bloomconstruct, x) = 1] = 1$
        \item Soundness: For any $x \notin S$: $\Pr[\bloomquery(\bloomconstruct(S), x) = 1] \leq \varepsilon$
        \item Element Permanence~\cite{filic1}: For any $x \in U$ and any $M$ such that $\bloomquery(M, x) = 1$, if $M^{\prime}$ is a later state after any sequence of insertions, it must hold that $\bloomquery(M^{\prime}, x) = 1$.
        \item Non-decreasing membership probability~\cite{filic1}: For any $x \in U$ and any $M$, let $M^{\prime} = \bloominsert(M, x)$. For all $y \in U$, it must hold that $\Pr[\bloomquery(M^{\prime}, y)] \geq \Pr[\bloomquery(M, y)]$.
    \end{enumerate}
    where the probabilities are taken over the random coins of $\bloomconstruct, \bloomquery$, and $\bloominsert$.
\end{definition}
\filic{}'s model uses a simulation-based definition for adversarial correctness. Their adversary $\adv = (\advconstruct, \advquery)$ has two components similar to Naor's model. $\advconstruct$ chooses any set $S \subseteq U$. $\advquery$ takes set $S$ and performs both adaptive queries and adaptive insertions to a Standard Bloom filter $\standardbloom$. $\advquery$ is allowed oracle access to $\bloomquery(M, \cdot)$ and $\bloominsert(M, \cdot)$. $\advquery$ is also allowed access to an oracle $\oracle_{M}$ that returns the internal representation $M$ of $\standardbloom$. We first discuss \filic{} et al.'s ideal simulator and then discuss their adversarial correctness notion.

\subsubsection{Ideal Simulator}

\filic{} show that Standard Bloom filter constructions have two properties, function decomposability and reinsertion invariance, that can be used to reason about their performance in an honest setting without having to specify a particular input distribution\footnote{While the simulator's behavior is described below, these two properties provide the theoretical foundation for why such non-adversarial simulation is possible. See Section 3 of ~\cite{filic1} for a detailed treatment.}. Let $\standardbloom$ be a Standard Bloom filter construction initialized with an empty set. Its behavior under an honest setting can be modelled using an algorithm $n$-NAI-gen that uniformly randomly samples $n$ unique elements from $U$, inserts them into $\standardbloom$, and returns the final representation $M$ of $\standardbloom$ after all $n$ insertions.

In the ideal world, $\adv$ interacts with a simulator $\simulator$ which provides a non-adversarially-influenced view of $\standardbloom$'s behavior. $\simulator$ maintains its own internal state which contains
\begin{enumerate}
    \item A representation $M$, which is bit string of length $m$ just like a Standard Bloom filter.
    \item A truly random function $f$ which maps any element $x \in U$ to $k$ indices in $[m]$.
    \item Lists \texttt{inserted} and \texttt{FPList}, for elements confirmed to be inserted and elements identified as false positives respectively.
    \item An integer counter, \texttt{ctr}, that keeps track of distinct insertions.
\end{enumerate}
$\simulator$ implements oracles for $\bloomconstruct, \bloomquery(M, \cdot), \bloominsert(M, \cdot),$ and $\oracle_{M}$ in the following way. When given a set $S \subseteq U$, $\bloomconstruct$ computes $k$ indices $f(x)$ for each $x \in S$, and sets the bit corresponding to each index in $M$ to $1$. It also adds $x$ to the \texttt{inserted} list and increments \texttt{ctr} for each $x$. When given an element $x \in U$ as input, $\bloomquery$ returns $1$ if $x$ is in \texttt{inserted} or \texttt{FPList}. Otherwise, it samples $k$ indices uniformly randomly from $[m]$ (it disregards $x$). If all $k$ indices in $M$ are set to $1$, $\bloomquery$ adds $x$ to \texttt{FPList} and returns $1$. Otherwise, it returns $0$. When given an element $x \in U$ as input, $\bloominsert$ does nothing if $x$ is in \texttt{inserted}. Otherwise, $\bloominsert$ updates $M$ by setting all the bits corresponding to the $k$ indices returned by $f(x)$ to $1$. It then adds $x$ to \texttt{inserted} and increments \texttt{ctr}. $\oracle_{M}$ simply returns $M$.

Since $\simulator$ queries on random indices instead of the given element $x$, $\simulator$'s response only reflects the underlying density of the bit string $M$. This is precisely the false positive probability under an honest setting.

\begin{figure*}[t!] 
    \centering
    \begin{minipage}[t]{0.65\textwidth}
        \begin{algorithm}[H] 
            \caption*{\(\texttt{Real}(\adv, \simulator, \dist)\)}
            \label{alg:real}
            \begin{algorithmic}[1] 
                \State Adversary $\advconstruct$ returns a set $S \subseteq U$.
                \State $\bloomconstruct$ takes $S$ and builds representation $M$.
                \State $out \sample \advquery^{\oracle}(S)$ where $\oracle = \{ \bloomquery(M, \cdot), \bloominsert(M, \cdot), \oracle_{M}\}$
                \State $d \sample \dist(out)$
                \State \Return $d$.
            \end{algorithmic}
        \end{algorithm}
    \end{minipage}
    \hfill
    \begin{minipage}[t]{0.32\textwidth}
        \begin{algorithm}[H] 
            \caption*{\(\texttt{Ideal}(\adv, \simulator)\)}
            \label{alg:ideal}
            \begin{algorithmic}[1] 
                \State $out \sample \simulator(\adv)$
                \State $d \sample \dist(out)$
                \State \Return $d$.
            \end{algorithmic}
        \end{algorithm}
    \end{minipage}
    \caption{Real and Ideal experiments in the \filic{} model.}
    \label{fig:security_experiments}
\end{figure*}

\subsubsection{Adversarial Correctness Notion}

In the \filic{} adversarial model, a security experiment a bit $b$ is flipped and based on the output, the adversary $\adv$ plays in either the real world ($b = 0$) or the ideal world ($b = 1$). After its interactions with either world are complete, $\adv$ must return an output $out$ that is given to a distinguisher $\dist$. The experiment then returns $\dist$'s output. In the ideal world, $\adv$ interacts with the ideal simulator defined above. In the real world, it is given oracle access to the algorithms of a real Standard Bloom filter construction $\standardbloom$. See Figure~\ref{fig:security_experiments} for the real and ideal experiments.

\filic{} et al.'s adversarial correctness notion is a bound on the distinguisher $\dist$'s probability of distinguishing between the real and the ideal world. To clearly distinguish this adversarial correctness notion from Naor et al.'s adversarial correctness notions, we will refer to it as \filic{} correctness.

\begin{definition}
    Let $\bloom$ be an insertable Bloom filter. $\bloom$ is $(q_{u}, q_{t}, q_{v}, t_{a}, t_{s}, t_{d}, \varepsilon)$-\filic{}-correct if for all adversaries $\adv$ running in time at most $t_{a}$, and making at most $q_{u}, q_{t}, q_{v}$ queries to the oracle for $\bloominsert(M, \cdot)$, the oracle for $\bloomquery(M, \cdot)$, and $\oracle_{M}$ respectively with an ideal simulator $\simulator$ that runs in time at most $t_{s}$, and for all distinguishers $\dist$ running in time at most $t_{d}$, we have
    \[
    \left|\Pr[\texttt{Real}(\adv, \simulator, \dist) = 1] - \Pr[\texttt{Ideal}(\adv, \simulator, \dist) = 1]\right| \leq \varepsilon
    \]
\end{definition}

\subsection{Learned Bloom filters}\label{sec:learnedbloomfilter}

Almashaqbeh, Bishop, and Tirmazi~\cite{almashaqbeh2024adversary} extend the NOY model to create an adversarial model for learned Bloom filters. We will refer to this as the Almashaqbeh-Bishop-Tirmazi (ABT) model in this paper. A learned Bloom filter is a Bloom filter that is working in collaboration with a learning model acting as a pre-filter. In this context, a regular Bloom filter, i.e., one that is not \textit{learned} is referred to as a \textit{classical} Bloom filter. Learned Bloom filters reduce the false positive rate of a Classical Bloom filter while maintaining the guarantee of no false negatives. A learned Bloom filter $\learnedbloom$ trains its learning model over the dataset $\learnedbloom$ represents, such that the model determines a function $\mathscr{L}$ that models this set. On input $x \in U$, $\mathscr{L}$ outputs the probability that $x \in S$, where $S$ is the input set. Relevant definitions from the ABT model are stated below.

\begin{definition}\label{def:trainingdataset}
Let $S \subseteq U$ be any set encoded by a Bloom filter. For any two sets $P \subseteq S$ and $N \subseteq U \setminus S$, the training dataset is the set $\mathscr{T} = \{ (x_{i}, y_{i} = 1) \mid x_{i} \in P \} \cup \{ (x_{i}, y_{i} = 0) \mid x_{i} \in N\}$.~\cite{almashaqbeh2024adversary}
\end{definition}

\begin{definition}\label{def:learning_model} 
  For an $\mathscr{L}: U \mapsto [0, 1]$ and threshold $\tau$, we say $\mathscr{L}$ is an $(S, \tau, \varepsilon_{p}, \varepsilon_{n})$-learning model, if for any set $S \subseteq U$ the following two properties hold:
  \begin{enumerate}
  \vspace{-5pt}
      \item P-Soundness: $\forall x \notin S : \Pr[\mathscr{L}(x) \geq \tau] \leq \varepsilon_{p}$
      \item N-Soundness: $\forall x \in S : \Pr[\mathscr{L}(x) < \tau] \leq \varepsilon_{n}$
  \end{enumerate}

  \noindent where the probability is taken over the random coins of $\mathscr{L}$.~\cite{almashaqbeh2024adversary}
\end{definition}

In the ABT model, a learned Bloom filter is defined in the following way. Similar to Naor and Oved~\cite{naor2022bet}, Almashaqbeh et al. only consider steady Bloom filters in which the query algorithm $\bloomquery$ does not change either the classical representation $M$ or the learned representation $(\mathscr{L}, \tau)$ of the input set $S$. 

\begin{definition}\label{def:learnedbloomfilter} 
A learned Bloom filter $\learnedbloom = (\mathbf{B}_{1}, \mathbf{B}_{2}, \mathbf{B}_{3}, \mathbf{B}_{4})$ is a 4-tuple of $\ppt$ algorithms: $\mathbf{B}_1$ is a construction algorithm, $\mathbf{B}_2$ is a query algorithm, $\mathbf{B}_{3}$ is a randomized algorithm that takes a set $S \subseteq U$ as input and outputs a training dataset $\mathscr{T}$, and $\mathbf{B}_{4}$ is a randomized algorithm that takes the training dataset $\mathscr{T}$ as input and returns a learning model $\mathscr{L}$ and a threshold $\tau \in [0, 1]$. The internal representation of $\learnedbloom$ contains two components: the classical component $M$ and the learned component ($\mathscr{L}, \tau)$. $\mathbf{B}_{2}$ takes as inputs an element $x \in U$, $M$, and ($\mathscr{L}, \tau$), and outputs 1 indicating that $x \in S$ and 0 otherwise. We say that $\mathbf{B}$ is an $(n, \tau, \varepsilon, \varepsilon_{p}, \varepsilon_{n})$-learned Bloom filter if for all sets $S \subseteq U$ of cardinality $n$, it holds that 

\begin{enumerate}
\vspace{-5pt}
    \item Completeness: $\forall{x} \in S : \Pr[\mathbf{B}_{2}(\mathbf{B}_{1}(S), \mathbf{B}_{4}(S, \mathbf{B}_{3}(S)), x) = 1] = 1$.
    \item Filter soundness: $\forall x \notin S : \Pr[\mathbf{B}_{2}(\mathbf{B}_{1}(S), \mathbf{B}_{4}(S, \mathbf{B}_{3}(S)), x) = 1] \leq \varepsilon$. 
    \item Learning model soundness: $\mathbf{B}_{4}(S, \mathbf{B}_{3}(S))$ is an $(S, \tau, \varepsilon_{p}, \varepsilon_{n})$-learning model.
\end{enumerate}
\noindent where the probabilities are over the random coins of $\mathbf{B}_{1}$, $\mathbf{B}_{3}$, and $\mathbf{B}_{4}$.~\cite{almashaqbeh2024adversary}
\end{definition}

Almashaqbeh et al. extend Naor et al.'s AB-test and BP-test to create versions suitable for learned Bloom filters. They refer to the learned versions of these tests as Learned-Always-Bet (LAB) and Learned-Bet-or-Pass (LBP), respectively. Almashaqbeh et al. introduce two learned Bloom filter constructions that are robust under LAB and LBP, respectively. Their constructions combine Naor et al.'s classical Bloom filter constructions with partitioned learned Bloom filters~\cite{plbf}. We will refer to these constructions as the ABT filter and the ABT Cuckoo filter.

\subsection{Randomized Response}\label{sec:randomizedresponse}

Warner's randomized response is one of the most common private set membership techniques, first proposed in $1965$~\cite{warner1965randomized}. Scientists have used it to survey set membership among a population for things that individual members wish to retain confidentiality about. A commonly used example is ``Are you a member of the Communist Party?''~\cite{hox_mulders}. Other examples include surveying the number of abortion recipients~\cite{abernathy1970estimates} and surveys regarding sexual orientation~\cite{xiangyu2014randomized}. In Warner's randomized response, the respondent answers a Yes/No question truthfully with probability $p$. With probability $1 - p$, the respondent flips a fair coin. The respondent answers \textit{Yes} if the coin is heads, and \textit{No} otherwise.  Thanks to this technique each respondent has \textit{plausible deniability} regarding their membership. Mangat~\cite{mangat1994improved} proposed a variant of Warner's randomized response. In Mangat's randomized response, a respondent answers truthfully to a Yes/No question with probability $p$. With probability $1 - p$, the respondent always answers Yes. Unlike Warner's randomized response, Mangat's variant only introduces one-sided error into the dataset.



%% file: sections/problems.tex
\section{Open Problems}\label{sec:openproblems}

Naor et al.~\cite {naor2022bet,lotan2025adversarially} introduce a hierarchy of game-based security notions for Bloom filters, as we discussed in Section~\ref{sec:preliminaries}. Separately, \filic{} et al.~\cite{filic1,filic2025deletions} introduce an alternate simulator-based security definition. \filic{} et al.'s security notion guarantees that the false positive rate observed by an adversary has a low probability of being significantly larger than the false positive rate observed in a non-adaptive setting.

Both Naor et al. and \filic{} et al.'s security notions share the same intuitive goal, i.e, minimizing false positives. However, no formal connection is currently known between them. Naor and Lotan~\cite{lotan2025adversarially} leave this as an open direction in their paper. Understanding the connection between these two approaches would help unify the robustness literature and clarify which notions provide stronger guarantees in practice.

\begin{problem}\label{prob:notion-unification}
    There are two dominant security frameworks for Bloom filters: Naor et al.'s game-based notions and Filic et al.'s simulator-based notion. Are there provable connections between the two frameworks?
\end{problem}

Another key distinction between the Naor and \filic{} models is that the NOY model does not allow the adversary to insert elements into a Bloom filter after construction, while the \filic{} model does allow insertions. Is it possible to create \textit{dynamic} versions of Naor et al.'s security notions, i.e., the AB-test and BP-test, etc., that allow insertions? For example, enabling an adversary to interleave insertions and queries before betting?

\begin{problem}\label{prob:naor-insertion}
     Can Naor et al.'s security games be generalized to insertable Bloom filters? 
\end{problem}

Almashaqbeh et al.~\cite{almashaqbeh2024adversary} extend the NOY model to support learned Bloom filters. They also propose learned Bloom filter constructions that are robust under learned versions of the AB-test and BP-test, respectively. However, it is unknown whether Almashaqbeh et al.'s robust learned Bloom filter constructions also satisfy \filic{} correctness. It is still undetermined whether \filic{} et al.'s adversarial model is compatible with learned Bloom filters and if secure constructions exist that satisfy the \filic{} correctness notion.

\begin{problem}
    Can the \filic{} adversarial model be extended to learned Bloom filters? Do learned Bloom filter constructions exist that satisfy the notion of \filic{} correctness?
\end{problem}

Similar to the classical Bloom filter constructions of Naor et al.~\cite{moni1,naor2022bet,lotan2025adversarially}, the learned Bloom filter constructions of Almashaqbeh et al. do not allow an adversary to insert elements into the Bloom filter after construction. Therefore, just like for Naor et al.'s classical Bloom filter constructions, further work is required to understand whether Almashaqbeh et al.'s learned Bloom filter constructions support any robustness notions for insertable Bloom filters. In fact, whether or not there exists a robust (under any notion) learned Bloom filter that allows inserts after construction is itself an unsolved problem.

\begin{problem}
    Are there any insertable learned Bloom filter construction that provably satisfies a meaningful robustness notion?
\end{problem}

Naor and Oved~\cite{naor2022bet} raise a question regarding the BP-test resilience of a Standard Bloom filter that uses keyed pseudo-random functions $F_{i}$ instead of public hash functions $h_{i}$. They write ``Note that it is not known whether replacing the hash functions with a PRF in the
standard construction of Bloom filters (i.e., the one in the style of Bloom’s original
one ~\cite{bloom1970}) results in a Bloom filter that is BP test resilient.''~\cite{naor2022bet}. This is an important question because the Standard Bloom filter construction is still one of the most widely deployed Bloom filter constructions. For example, it is deployed in the Linux Kernel~\cite{kernelbpfmap} and Google's LevelDB~\cite{leveldbbloom}. We precisely define the problem below.

\begin{problem}\label{prob:prf-standard-bloom}
    Let $\modifiedbloom$ be a modified construction of a Standard Bloom filter $\standardbloom$ that replaces each hash function $h_{i}$ in $\standardbloom$ with a keyed PRF $F_{i}$. Does $\modifiedbloom$ satisfy $(n, \varepsilon)$-strong resilience under the BP-test?
\end{problem}  
As we discussed in the introduction, many works~\cite{sengupta2017sampling,reviriego2022privacy,galan2022privacy} have shown that Bloom filters leak information regarding the stored input set. An open problem is exploring rigorous privacy guarantees based on the notion of differential privacy for Bloom filters.

\begin{problem}\label{prob:bloom-privacy}
    Are there any Bloom filter constructions that provide rigorous differential privacy guarantees for the set they store?
\end{problem}

The robustness tests under the NOY model, including the AB-test and the BP-test, assume that the adversary makes \textit{distinct} queries. Lotan and Naor~\cite{lotan2025adversarially} (and Naor and Oved~\cite{naor2022bet} but with less details) pose an open problem regarding the robustness of Bloom filters when query repetition is allowed. Lotan and Naor's proposed direction can be broken down into three precise questions.

\begin{problem}\label{prob:repeated-query-bp-test}
    Does there exist a Bloom filter construction that satisfies BP-test resilience when the adversary is allowed to repeat queries?
\end{problem}

Note that this requires extending the BP-test definition to handle repeated queries instead of only allowing distinct queries. Bender et al.~\cite{bender2018bloom} introduce a construction called a broom filter that has provable guarantees under repeated queries. In their paper, Bender et al. introduce their adversarial model for Bloom filters, which we will refer to as the Bender model. Unlike the NOY model, which only allows distinct queries, the Bender model allows repeated queries. Lotan and Naor also ask whether the NOY model is compatible with the Bender model or has provable connections.

\begin{problem}\label{prob:bender-naor-unification}
    Are there any provable connections between Naor et al.'s security notions, which do not allow query repetition, and Bender et al.'s adversarial model, which does allow query repetition?
\end{problem}

Mitzenmacher et al.~\cite{mitzenmacher2020adaptive} provide a Bloom filter construction called an Adaptive Cuckoo Filter that removes false positives after they are queried. However, it is not known whether Adaptive Cuckoo Filters are provably adaptive~\cite{bender2018bloom} under any of the discussed adversarial models. 

\begin{problem}
    Are there any provable bounds on the adversarial robustness of Adaptive Cuckoo Filters under a known Bloom filter adversarial model?
\end{problem}

Finally, there is another well-known adversarial model for Bloom filters, introduced by Clayton, Patton, and Shrimpton~\cite{claytonetal}. The Clayton-Patton-Shrimpton (CPS) model extends the NOY model using a similar game-based formalism. Recall that Naor et al.'s AB-test allows an adversary to make $t$ distinct adaptive queries, before outputting a new, unseen challenge query. Only this challenge query needs to be a false positive for the adversary to win. Clayton et al. instead allow an adversary to win if the adversary can forge a number of distinct false positive queries during its entire execution that is above a parametrized threshold. A formal reduction or separation between the CPS and NOY models would clarify their comparative strengths and applicability across Bloom filters.

\begin{problem}\label{prob:cps-naor-unification}
    Are there any provable connections between the NOY model and the CPS model?
\end{problem}

\subsection{Taxonomy}

We present a taxonomy of the open problems we discussed that unifies the contributions of recent work on the adversarial robustness of Bloom filters. Our classification contains three axes: robustness notions, construction features, and model relationships. We provide a table for each axis, with open problems indicated with $\diamond$ in the tables.\medskip

\noindent\textbf{Robustness Notions}: provable guarantees for Bloom filters differ significantly across definitions, with two dominant families of adversarial models. The first family encompasses game-based notions, including those that cover learning-based robustness. Naor et al.~\cite{moni1,naor2022bet,lotan2025adversarially}, Clayton et al.~\cite{claytonetal}, and Almashaqbeh et al.~\cite{almashaqbeh2024adversary} define adversarial correctness in terms of win conditions in an interactive game. The second family relies on simulator-based notions. The only current notable example of this is the work of \filic{} et al.~\cite{filic1,filic2025deletions}. Privacy-based notions for Bloom filters remain less well-explored. \filic{} et al.~\cite{filic1} propose a simulator-based privacy definition for Bloom filters based on information leakage. We summarize prior work in terms of this axis in Table~\ref{tab:security_summary}.\medskip

\noindent\textbf{Construction features}: we map known secure Bloom filter constructions to the robustness notions they satisfy in Table~\ref{tab:constructions_summary}. This includes constructions that are learned or classical, and insertable or static (i.e., no post-construction updates). The majority of open problems on this axis relate to extending robustness guarantees to learned, insertable, or repeated query settings.\medskip

\noindent\textbf{Model relationships}: Table~\ref{tab:models_summary} summarizes the space of known and unknown connections between adversarial models. Note that this does not include connections between notions within the same model, such as those explored in the work of Naor and Oved~\cite{naor2022bet} or Almashaqbeh et al.~\cite {almashaqbeh2024adversary}. The vision here is for the community to incrementally develop a single unified adversarial model that captures all security notions.

\begin{table}[t!]
\centering
\caption{Mapping of robustness notions to Bloom filter classes}
\label{tab:security_summary}
\renewcommand{\arraystretch}{1.2}
\begin{tabularx}{\textwidth}{@{}l *{4}{C} @{}}
\toprule
\textbf{Notion} & \textbf{Classical Bloom Filter} & \textbf{Learned Bloom Filter} & \textbf{Insertable Bloom Filter} & \textbf{Repeated Queries} \\
\midrule
NOY AB             & \cite{naor2022bet} & $\diamond$ & $\diamond$ & $\diamond$ \\
NOY BP             & \cite{naor2022bet} & $\diamond$ & $\diamond$ & $\diamond$ \\
\filic{}     & \cite{filic1}      & $\diamond$ & \cite{filic1} & $\diamond$ \\
ABT LAB & $\diamond$ & \cite{almashaqbeh2024adversary} & $\diamond$ & $\diamond$ \\
ABT LBP & $\diamond$ & \cite{almashaqbeh2024adversary} & $\diamond$ & $\diamond$ \\
Bender             & \cite{bender2018bloom} & $\diamond$ & \cite{bender2018bloom} & \cite{bender2018bloom} \\
Diff. Privacy     & $\checkmark$        & $\checkmark$ & $\diamond$ & $\checkmark$ \\
\bottomrule
\end{tabularx}

\vspace{0.5em}
\footnotesize
\textbf{Legend:} $\checkmark$ = Contribution of this paper; $\diamond$ = Open problem.
\end{table}

\begin{table}[t!]
\centering
\caption{Robustness and privacy notions satisfied by each relevant Bloom filter construction.}
\label{tab:constructions_summary}
\renewcommand{\arraystretch}{1.1}
\small  
\begin{tabularx}{\textwidth}{@{}l *{7}{C} @{}}
\toprule
\textbf{Construction} & \textbf{Learned} & \textbf{Insert.} & \textbf{Naor AB} & \textbf{Naor BP} & \textbf{\filic{}} & \textbf{Rep. Queries} & \textbf{Diff. Priv.} \\
\midrule
SBF & N & Y & N & N & N & N & N \\
PRF-Backed SBF & N & N & $\diamond$ & $\times$ & $\diamond$ & $\diamond$ & $\diamond$ \\
NY & N & N & Y & N & $\diamond$ & $\diamond$ & $\diamond$ \\
NOY Cuckoo & N & N & Y & Y & $\diamond$ & $\diamond$ & $\diamond$ \\
FPUV & N & N & $\checkmark$ & $\diamond$ & Y & $\diamond$ & $\diamond$ \\
ABT & Y & N & Y & N & $\diamond$ & $\diamond$ & $\diamond$ \\
ABT Cuckoo & Y & N & Y & Y & $\diamond$ & $\diamond$ & $\diamond$ \\
Broom & N & Y & $\diamond$ & $\diamond$ & $\diamond$ & Y & $\diamond$ \\
Mangat & $\checkmark$ & N & $\diamond$ & $\diamond$ & $\diamond$ & $\diamond$ & $\checkmark$\\
Warner & $\checkmark$ & N & $\diamond$ & $\diamond$ & $\diamond$ & $\diamond$ & $\checkmark$\\
\bottomrule
\end{tabularx}

\vspace{0.5em}
\footnotesize
\textbf{Legend:} Y/N = yes / no result by prior work; $\checkmark$ / $\times$ = yes / no result contributed by this paper; $\diamond$ = open problem.
\end{table}

\begin{table}[t!]
\centering
\caption{Summary of known connections between adversarial models}
\label{tab:models_summary}
\renewcommand{\arraystretch}{1.1}
\small  
\begin{tabularx}{\textwidth}{@{}l *{7}{C} @{}}
\toprule
\textbf{Model} & \textbf{NOY} & \textbf{\filic{}} & \textbf{ABT} & \textbf{Bender} & \textbf{CPS}\\
\midrule
\textbf{NOY} & $\ast$ & $\checkmark$ & $\diamond$ & $\diamond$ & $\diamond$\\
\textbf{\filic{}} & $\checkmark$ & $\ast$ & $\diamond$ & $\diamond$ & $\diamond$\\
\textbf{ABT} & $\diamond$ & $\diamond$ & $\ast$ & $\diamond$ & $\diamond$\\
\textbf{Bender} & $\diamond$ & $\diamond$ & $\diamond$ & $\ast$ & $\diamond$\\
\textbf{CPS} & $\diamond$ & $\diamond$ & $\diamond$ & $\diamond$ & $\ast$\\
\bottomrule
\end{tabularx}

\vspace{0.5em}
\footnotesize
\textbf{Legend:} $\checkmark$ = provable connections contributed by this paper; $\diamond$ = open problem; $\ast$ = trivially true.
\end{table}

%% file: sections/privacy.tex
\section{Private Bloom filters}\label{sec:privatebloomfilters}

In this section, we attempt to solve Problem~\ref{prob:bloom-privacy} by providing two constructions for Bloom filters with differential privacy guarantees. We first discuss how to adapt differential privacy~\cite{dwork2014algorithmic} and asymmetric differential privacy~\cite{takagi2022asymmetric} for unordered sets and show that Mangat's randomized response (Section~\ref{sec:randomizedresponse}) satisfies asymmetric differential privacy. We discuss two Private Bloom filter constructions, the Mangat filter and the Warner filter. We then investigate the error rates of the private Bloom filter constructions as compared to the Standard Bloom filter construction. Finally, we discuss how the private Bloom filter constructions are also applicable to learned Bloom filters.

\subsection{Differential Privacy on Unordered Sets}

For any two sets $A, B$, we can use an unweighted version of the Jaccard distance, $d_{sj}(A, B) = | A \cup B | - | A \cap B |$, to measure set similarity. The well-known notions of symmetric~\cite{dwork2014algorithmic} and asymmetric~\cite{takagi2022asymmetric} differential privacy can then be written in terms of unordered sets in the following way.

\begin{definition}\label{def:set_privacy} A randomized algorithm $\randomalg$ satisfies $(\epsilon, \delta)$-differential privacy~\cite{dwork2014algorithmic} if for any two sets $S, S^{\prime}$ s.t $d_{sj}(S, S^{\prime}) \leq 1$ and for any possible output range $O \subseteq \text{Range}(\randomalg)$,
\[
P[\randomalg(S) \in O] \leq e^{\epsilon} P[\randomalg(S^{\prime} \in O] + \delta
\]
where the probabilities are over $\randomalg$.
\end{definition}

\begin{definition}\label{def:asymmetric_set_privacy} 
A randomized algorithm $\randomalg$ satisfies $(\varepsilon, \varepsilon^{\prime}, \delta)$-asymmetric differential privacy~\cite{takagi2022asymmetric} if for any two sets $S, S^{\prime}$ such that $d_{sj}(S, S^{\prime}) \leq 1$, and for any possible output range $O \subseteq \text{Range}(\randomalg)$, the following two properties hold:

\begin{enumerate}

\item If $S^{\prime} = S \setminus \{x\}$ for some $x \in S$, then
\(
\Pr[\randomalg(S) \in O] \leq e^{\varepsilon} \Pr[\randomalg(S^{\prime}) \in O] + \delta
\)

\item If $S^{\prime} = S \cup \{x\}$ for some $x \not\in S$, then
\(
\Pr[\randomalg(S) \in O] \leq e^{\varepsilon^{\prime}} \Pr[\randomalg(S^{\prime}) \in O] + \delta
\)
\end{enumerate}

\noindent where probabilities are over $\randomalg$.
\end{definition}

Asymmetric differential privacy aligns well with many known set membership scenarios where presence in a set is sensitive, while absence is not. For example, knowing that an individual belongs to the set of Communist party members, HIV patients, or recipients of abortions can be highly sensitive, whereas knowing that an individual is not in these sets often does not reveal sensitive information. There are also situations where this type of privacy guarantee is \textit{necessary}. For example, in epidemic analysis, when creating a set of the number of infected individuals that visited a location~\cite{takagi2022asymmetric}, having a two-sided error is not useful.

There is a well-known result~\cite{dwork2014algorithmic} that demonstrates the differential privacy of Warner's randomized response. When the probability of the respondent answering a question truthfully is $p$, Warner's randomized response satisfies \( (\ln \left( \frac{p}{1 - p} \right), 0) \)-differential privacy. This result also holds for differential privacy when applied to sets. We now show that Mangat's randomized response satisfies asymmetric differential privacy for sets. This will be needed for constructing private Bloom filters.

\begin{theorem}\label{thm:mangat}
    Mangat's randomized response satisfies $(\ln(\frac{1}{1 - p}), \ln({1 - p}), 0)$-asymmetric differential privacy.
\end{theorem}

\begin{proof}
    Let $S, S^{\prime}$ be two sets s.t $d_{sj}(S, S^{\prime}) = 1$, and $\randomalg$ be Mangat's randomized response. The probabilities are taken over $\randomalg$. First, take the case where $S^{\prime} \subset S$, i.e., $S^{\prime} = S \setminus \{x\}$ for some $x \in S$.
    \[
        \frac{\Pr[\randomalg(S) \in O]}{\Pr[\randomalg(S^{\prime}) \in O]} = \frac{ \Pr[x \in \randomalg(S)]}{\Pr[x \in \randomalg(S^{\prime})]} = \frac{1}{1 - p}
    \]
    Hence, $\varepsilon = \ln (\frac{1}{1 - p})$. Now take the case where $S \subset S^{\prime}$, i.e., $S^{\prime} = S \cup \{ x \}$ for some $x \notin S$.
    
    \[
        \frac{\Pr[\randomalg(S) \in O]}{\Pr[\randomalg(S^{\prime}) \in O]} = \frac{ \Pr[x \in \randomalg(S)]}{\Pr[x \in \randomalg(S^{\prime})]} = 1 - p
    \]
    Therefore $\varepsilon^{\prime} = \ln{(1 - p)}$ and $\delta = 0$. The result follows.
\end{proof}

\subsection{Mangat and Warner filters}

Since Bloom filters execute randomized algorithms to store sets, the set privacy notions can be modified to get analogous Bloom filter privacy notions.  

\begin{definition}\label{def:bf_privacy} An $(n, \varepsilon)$-Bloom filter $\bloom = (\bloomconstruct, \bloomquery)$ is an $(n, \varepsilon, \varepsilon_{p}, \delta_{p})$-private Bloom filter if for all $S, S^{\prime} \subseteq U$ such that $d_{sj}(S, S^{\prime}) \leq 1$ and for all representations $M$, \(\Pr[\bloomconstruct(S) = M ] \leq e^{\varepsilon_{p}} P[\bloomconstruct(\S^{\prime}) = M ] + \delta_{p}\) where the probabilities are over the coins of $C_{r}$.
\end{definition}
An $(n, \varepsilon, \varepsilon_{p}, \varepsilon_{p}^{\prime}, \delta_{p})$-asymmetric private Bloom filter can be defined analogously using the asymmetric differential privacy definition for sets.

We now introduce two private Bloom filter constructions, Mangat filters and Warner filters, based on Mangat and Warner randomized response respectively. Mangat filters keep a Bloom filter's traditional one-sided guarantees, i.e., no false negatives only false positives. However, Mangat filters only satisfy asymmetric differential privacy. Warner filters satisfy (symmetric) differential privacy, at the cost of returning false negatives with a small probability.

A Mangat filter $\bloommangat$ can be constructed from any Bloom filter $\bloom = (\bloomconstruct, \bloomquery)$ in the following way. Replace $\bloomconstruct$ with $\bloomconstruct^{\prime}$ that works in the following way. Fix a probability $p \in (0, 1)$. Initialize $S^{\prime} \gets S$. For each $x \in U \setminus S$, add $x$ to $S^{\prime}$ with probability $1 - p$. Call the original construction algorithm $\bloomconstruct$ on $S^{\prime}$ instead of $S$, i.e., return $\bloomconstruct(S^{\prime})$. The query algorithm $\bloomquery$ remains unchanged.

\begin{theorem}\label{thm:mangatbloomfilter} Mangat filter satisfies $(\ln(\frac{1}{1 - p}), \ln({1 - p}), 0)$-asymmetric differential privacy.
\end{theorem}

\begin{proof}
    A Mangat filter is a special case of Mangat's randomized response mechanism applied to set membership. Let $S$ and $S^{\prime}$ be two sets s.t $d_{sj}(S, S^{\prime}) \leq 1$. A Mangat filter modifies input set \( S \) by adding each element \( x \in U \setminus S \) with probability \( p \). This follows the structure of Mangat's randomized response mechanism from Theorem~\ref{thm:mangat} and satisfies the given same asymmetric differential privacy guarantee.
\end{proof}

\begin{figure}[t!]
    \centering
    \includegraphics[width=0.6\linewidth]{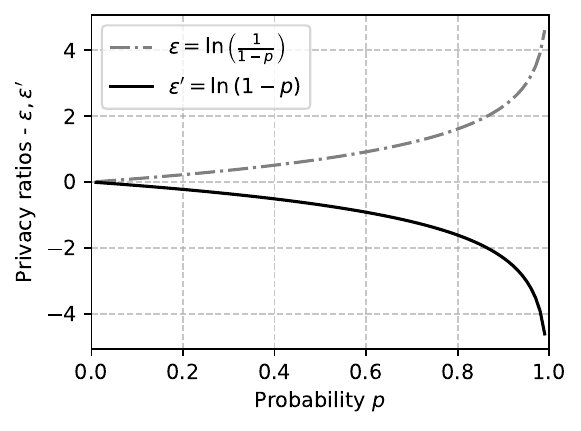}
    \caption{The asymmetric privacy bounds of a Mangat filter as $p$ varies. $\varepsilon$ is the privacy of an element not in the output set, while $\varepsilon^{\prime}$ is the privacy of an element in the output set.}
    \label{fig:privacy_ratios}
\end{figure}

Theorem~\ref{thm:mangatbloomfilter} proves that a Mangat filter satisfies \textit{asymmetric} differential privacy by adding elements with a controlled probability $1 - p$. Asymmetric differential privacy enables us to explicitly model scenarios where an adversary $\adv$'s ability to infer the presence of an element in the original set is significantly weaker than $\adv$'s ability to infer the absence of an element in the original set. We illustrate this in Figure~\ref{fig:privacy_ratios}. $\varepsilon$ and $\varepsilon^{\prime}$ model $\adv$'s ability to infer the absence and presence of an element in the original set, respectively. Since a Mangat filter never removes an element in the original set, $\varepsilon$ does not meaningfully constrain $A$’s ability to infer absence. $\varepsilon^{\prime}$, however, provides meaningful privacy for presence. The privacy increases as $\varepsilon^{\prime} \rightarrow 0$ which happens as $1 - p \rightarrow 1$, i.e., $p \rightarrow 0$. Intuitively, when all elements in the universe appear in the output set, $\adv$ has little probability of distinguishing which elements were in the original set. The privacy decreases as $\varepsilon^{\prime} \rightarrow -\infty$ ($1 - p \rightarrow 0$), i.e, as the Mangat filter probabilistically adds fewer elements. 

A Warner filter $\bloomwarner$ can be constructed for a Bloom filter by modifying its construction algorithm by replacing $\bloomconstruct$ with $\bloomconstruct^{\prime}$ that works as follows. Fix a probability $p \in (\frac{1}{2}, 1)$. Initialize $S^{\prime} \gets \emptyset$. For each $x \in S$, add $x$ to $S^{\prime}$ with probability $p$. For each $x \in U \setminus S$, add $x$ to $S^{\prime}$ with probability $1 - p$.

\begin{theorem}\label{thm:warnerbloomfilter} Warner filter satisfies \( (\ln \left( \frac{p}{1 - p} \right), 0) \)-differential privacy.
\end{theorem}

\begin{proof}
    A Warner filter is a special case of Warner's randomized response mechanism applied to set membership. Let $S$ and $S^{\prime}$ be two sets s.t $d_{sj}(S, S^{\prime}) \leq 1$. A Warner filter modifies input set \( S \) by removing each element $x \in S$ with probability $p$ and adding each element \( x \in U \setminus S \) with probability \( 1 - p \). This follows the structure of Warner's randomized response mechanism and therefore we can directly apply the well-known result~\cite{dwork2014algorithmic} that Warner's randomized response satisfies the given differential privacy guarantee.
\end{proof}

\subsection{Error Rate Analysis}

We now investigate how our method for adding privacy affects the false positive rate (FPR) and its false negative rate (FNR) of a given Standard Bloom filter $\standardbloom$. We do not classify queries to elements added by our privacy-preserving algorithms as false positives, i.e., a query on $x \in S^{\prime} \setminus S$ is not a false positive. These elements are not representative of typical false positives, which arise naturally due to the probabilistic nature of the $\standardbloom$. As such, we exclude these elements from the FPR calculations to focus on the inherent accuracy of the $\standardbloom$ under privacy-preserving conditions. This distinction ensures a clear separation between errors resulting from $\standardbloom$'s one-sided guarantees and those intentionally introduced for privacy purposes.

Assume $\standardbloom$ stores set $S \subseteq U$, has internal bit-string $M$, and $k$ hash functions. Let $\text{FPR}(S, M, k)$ and $\text{FNR}(S, M, k)$ be functions that return the expected FPR and expected FNR of $\standardbloom$, respectively. Then the FPR and FNR of a private Bloom filter built on top of $\standardbloom$ that constructs set $S^{\prime}$ from the original set $S$ will be $\text{FPR}(S^{\prime}, M, k)$ and $\text{FNR}(S^{\prime}, M, k)$ respectively. It is well-known that $\standardbloom$ has approximately the following false positive rate~\cite{MitzenmacherBroder2004}, $\textrm{FPR}(S, M, k) = (1 - e^{-k \cdot |S|/|M|})^{k}$ where $|M|$ is the length of the bit-string $M$. For a given set $S$, the expected cardinality of the set $S^{\prime}$ stored by a 
Mangat filter is $\left|S^{\prime}\right| = \left|S\right| +\, (1 - p) (\left|U\right| - \left|S\right|)$. Similarly, for a Warner filter, it is $|S^{\prime}| = |S| +\, p (|U| - |S|) -\, (1 - p)|S|$. We can replace $|S|$ in the FPR equation for $\standardbloom$ with these expressions to get FPR expressions for private Bloom filters. When using the Warner filter, we will also have a non-zero FNR, which is simply the probability that a given $x \in S$ is not included in the set $S^{\prime}$ by the construction algorithm, i.e, $1 - p$.

\subsection{Applicability to Learned Bloom filters}

Mangat and Warner filters modify the input set $S$ prior to Bloom filter construction, without altering the structure of the Bloom filter itself. As a result, both constructions are fully compatible with learned Bloom filters, including those modeled in the ABT framework~\cite{almashaqbeh2024adversary}. A learned Bloom filter $\learnedbloom = (\mathbf{B}_{1}, \mathbf{B}_{2}, \mathbf{B}_{3}, \mathbf{B}_{4})$ relies on a training dataset $\tau = \mathbf{B}_{3}(S)$ generated from the input set $S$. Since the Mangat and Warner filters perturb $S$ to produce a new set $S^{\prime}$, the learning model is trained on $\tau^{\prime} = \mathbf{B}_{3}(S^{\prime})$ instead of $\tau$. Mangat and Warner filters thus act as a privacy-preserving pre-processing step on the input set before training and construction.

In other words, if $\learnedbloom$'s underlying learning model and classical Bloom filter(s) satisfy standard correctness guarantees over $S^{\prime}$, and if $S^{\prime}$ is generated using a randomized response mechanism satisfying symmetric or asymmetric differential privacy, then the overall construction $\learnedbloom$ inherits the same privacy guarantees. No modification to the data structure and its underlying algorithms is required. This observation allows Warner and Mangat filters to serve as generic wrappers for private learned Bloom filters, expanding the scope of private Bloom filter constructions beyond classical Bloom filters. To our knowledge, this is the first approach that provides provable differential privacy guarantees for learned Bloom filters.

%% file: sections/bridging-models.tex
\section{Bridging NOY and \filic{} Models}\label{sec:bridgingnoyfilic}

\begin{figure}[t!]
\centering
\begin{tikzpicture}[
    node distance=1.0cm and 3.5cm,
    >=Stealth,
    text node/.style={draw, align=center, font=\bfseries, minimum width=1.8cm, minimum height=0.8cm},
    arrow label/.style={font=\small},
    implies/.style={->, double distance=1.5pt, shorten <=2pt, shorten >=2pt},
    notimplies/.style={->, double distance=1.5pt, shorten <=2pt, shorten >=2pt},
    cross/.style={draw=red, thick, line width=1pt}
  ]

  \node (filic)     [text node]                              {\texttt{\filic{} Correctness}};
  \node (abtest)         [text node, right=of filic]      {\texttt{AB Test}};
  \node (bptest) [text node, below=of abtest]      {\texttt{BP Test}};

  \draw [notimplies] ([xshift=20pt]filic.south) -- node[pos=0.3, above, sloped, arrow label] {}  node[pos=0.5, sloped, font=\huge\bfseries\color{red}] {$\mathbf{\times}$} ([yshift=6pt]bptest.west);

  \draw [notimplies] ([xshift=-8pt]abtest.south) -- node[pos=0.3, left, arrow label] {}  node[pos=0.5, sloped, font=\huge\bfseries\color{red}] {$\mathbf{\times}$} ([xshift=-8pt]bptest.north);

  \draw [implies] ([xshift=8pt]bptest.north) -- node[pos=0.3, right, arrow label] {}  ([xshift=8pt]abtest.south);

  \draw [notimplies] ([yshift=-3pt]bptest.west) -- node[pos=0.3, sloped, below, xshift=-5pt, font=\small] {} node[pos=0.3, sloped, font=\huge\bfseries\color{red}] {$\mathbf{\times}$} ([xshift=-5pt]filic.south);

  \draw [implies] ([yshift=5pt]filic.east) -- node[pos=0.3, above, arrow label] {} ([yshift=5pt]abtest.west);

  \draw [notimplies] ([yshift=-5pt]abtest.west) -- node[pos=0.3, below, arrow label] {} node[pos=0.3, sloped, font=\huge\bfseries\color{red}] {$\mathbf{\times}$} ([yshift=-5pt]filic.east);

\end{tikzpicture}
\caption{Connections between Naor and \filic{} notions.}
\label{fig:securityimplications}
\end{figure}

In this section, we solve Problem~\ref{prob:notion-unification} by providing provable connections between the robustness notions of Naor et al.'s model and \filic{} et al.'s model. Lotan and Naor~\cite{lotan2025adversarially} provide a counter-example demonstrating that \filic{} correctness does not imply resilience under the BP test. Our work shows that \filic{} correctness \textit{does} imply resilience under the AB test. We also use a counter-example construction introduced by Almashaqbeh et al.~\cite{almashaqbeh2024adversary} to demonstrate that resilience under AB test or BP test does not imply \filic{} correctness.

\begin{theorem}
    If a Bloom filter $\bloom$ is $(q_{u}, q_{t}, q_{v}, t_{a}, t_{s}, t_{d}, \varepsilon)$-\filic{}-correct then $\bloom$ is also $(n, q_{t} - 1, 2\varepsilon)$-resilient under the AB test for adversaries running in time at most $t_{a}$.
\end{theorem}

\begin{proof}
    Assume $\bloom$ is not $(n, q_{t} - 1, 2\varepsilon)$-resilient under the AB test, i.e., there exists an adversary $\adv = (\advconstruct, \advquery)$ running in time at most $t_{a}$ who can win the AB test with probability non-negligibly greater than $\varepsilon$. We will show how to construct an adversary $\adv^{\prime}$ using $\adv$ and a distinguisher $\dist$ that distinguishes between the Real and Ideal worlds in the \filic{} experiment with probability greater than $\varepsilon$. 
    
    $\adv^{\prime}$ plays the experiments in the \filic{} model in the following way. $\adv^{\prime}$ runs $\advconstruct$ to get set $S$ which it forwards to the \filic{} experiment. $\advquery$ requires oracle access to $\bloomquery(M, \cdot)$ to make $q_{t} - 1$ adaptive queries. $\adv^{\prime}$ forwards $\advquery$'s queries to the $\bloomquery(M, \cdot)$ oracle provided to $\adv^{\prime}$ in the \filic{} experiment. After $q_{t} - 1$ adaptive queries, $\advquery$ returns $x^{*}$ as required by the AB test. $\adv^{\prime}$ uses the last query in its $q_{t}$ query budget for oracle $\bloomquery(M, \cdot)$ to get $out = \bloomquery(M, x^{*})$ and returns $out$ as its output. Distinguisher $\dist$ outputs $d = out$, i.e., it decides it is in the real world if $out = 1$ and in the ideal world otherwise.
    
    Since we assumed $\bloom$ is not $(n, q_{t} - 1, 2\varepsilon)$-resilient under the AB test, $\Pr[\texttt{Real}(\adv, \simulator, \dist) = 1] > 2\varepsilon + \negl$. 
    Since the ideal simulator, $\simulator$, ignores $\adv$'s output $x^{*}$ and picks $k$ indices uniformly randomly, $\adv$'s choice has no impact on $\simulator$'s false positive probability. $\simulator$'s false positive probability is $\bloom$ false positive probability in a non-adversarial setting, which is at most $\bloom$'s false positive probability in an adversarial setting, i.e, $\Pr[\texttt{Ideal}(\adv, \simulator, \dist) = 1] \leq \varepsilon$. Therefore \(\left|\Pr[\texttt{Real}(\adv, \simulator, \dist) = 1] - \Pr[\texttt{Ideal}(\adv, \simulator, \dist) = 1]\right| > |2\varepsilon + \negl - \varepsilon| > \varepsilon\), and therefore the distinguishing advantage is larger than $\varepsilon$ violating \filic{} correctness. Hence, we have shown that $\bloom$ is not $(n, q_{t} - 1, 2\varepsilon)$-resilient under the AB test it is also not $(q_{u}, q_{t}, q_{v}, t_{a}, t_{s}, t_{d}, \varepsilon)$-\filic{}-correct. The result follows.
\end{proof}
The converse does not hold, i.e., resilience under the AB test does not imply \filic{} correctness. A trivial counter-example is a NY filter, which is resilient under the AB test~\cite{naor2022bet}. A NY filter uses a Standard Bloom filter and a keyed pseudo-random permutation $\prp_{\sk}$ with secret key $\sk$. For any element $x \in U$, the NY filter stores and queries $\prp_{\sk}(x)$ instead of $x$ directly. Almashaqbeh et al.~\cite{almashaqbeh2024adversary} show that a modified NY filter that stores the secret key $\sk$ as part of its internal representation $M$ is still resilient under the AB test. However, such a construction will not satisfy \filic{} correctness as the \filic{} model allows oracle access to the internal representation $M$, allowing the adversary to read the secret key. The same argument can be used to show that resilience under the BP test also does not imply \filic{} correctness.

%% file: sections/prf-standard-bloom-filter.tex
\section{PRF-backed Standard Bloom filters}\label{sec:prfbackedsbfbptest}

In this section, we solve Problem~\ref{prob:prf-standard-bloom} by proving that a PRF-backed Standard Bloom filter construction is \textit{not} strongly resilient under the BP-test. We actually prove a stronger result showing that the construction is not strongly resilient even when using truly random functions instead of PRFs.

When encoding a set, $\modifiedbloom$ may get \textit{saturated}, i.e, every bit in $\modifiedbloom$'s representation $M \in \bin^{m}$ is set to $1$. Let $\psaturate$ be the saturation probability of $\modifiedbloom$. Finding $\psaturate$ is equivalent to solving the Coupon Collector's Problem~\cite{evilchoicesbloom}. The probability of a given bit being $0$ is $(1 - \frac{1}{m})^{nk}$. The probability of at least $1$ of $m$ bits being $0$ is $\leq m(1 - \frac{1}{m})^{nk} \leq me^{-nk/m}$, using a union bound and the inequality $1 + x \leq e^{x}$ for any $x \in \RR$~\cite{mitzenmacher2017probability}. The saturation probability of $\modifiedbloom$ is then $\psaturate \geq 1 - me^{-nk/m}$. 

\begin{theorem}\label{thm:prf-backed-standard-bloom}
    Let $\modifiedbloom$ be a modified construction of a Standard Bloom filter $\standardbloom$ that replaces each hash function $h_{i}$ in $\standardbloom$ with a truly random function $f_{i}$. Let $\varepsilon \in (0, 1)$ and $n \in \NN$. $\modifiedbloom$ is not $(n, t, \delta)$-resilient under the BP-test for any $t \in \NN$ and any $\delta \in (0, 1)$ such that $\delta < \psaturate$, where $\psaturate$ is $\modifiedbloom$'s saturation probability.
\end{theorem}

\begin{proof}
    Suppose adversary $\adv$ plays $\texttt{AdaptiveGame}_{\adv, t}(\lambda)$ using the following strategy. $\adv$ chooses any $S \subset U$ of cardinality $n$ and chooses $t$ elements $x_{i} \sample U \setminus S$. In the game, $\adv$ uses its $t$ allowed queries to $\modifiedbloom$ to query each $x_{i}$. If all $x_{i}$s are false positives, $\adv$ chooses to bet, and bets on an $x^{*} \sample U$ (We don't sample $x^{*} \sample U \setminus (S \cup \{x_{1}, \cdots, x_{t}\})$ here as the proof is simplified when $x^{*}$ is chosen independently of all $x_{i}$, and assuming $U$ is large the probability of $x^{*}$ not being distinct from all $x_{i}$ is negligible). Otherwise, $\adv$ passes. We show that the expected value of $\adv$'s profit $C_{\adv}$ is not negligible with this strategy. W.l.o.g., fix the size of the bit array $m$, the number of PRFs $k$, and the cardinality $n$ of the encoded set. 
    
    Let $\fpr$ be the false positive probability of $\modifiedbloom$, in expectation. Let $\fp(x^{*})$ denote whether or not $x^{*}$ is a false positive, $E_{b}$ be the event denoting $\adv$ betting (instead of passing), and $E_{s}$ be the event denoting the saturation of $\modifiedbloom$. If $\adv$ follows the given strategy, then 

    \begin{align*}
    \expect{C_{\adv}} &= \frac{1}{\delta} \Pr[\fp(x^{*}) \cap E_{b}] - \frac{1}{1 - \delta} \Pr[\lnot \fp(x^{*}) \cap E_{b}] + 0 \cdot \Pr[\lnot E_{b}]\\
    &= \frac{1}{\delta} \Pr[\fp(x^{*}) \mid E_{b}] \Pr[E_{b}] - \frac{1}{1 - \delta}\Pr[\lnot \fp(x^{*}) \mid E_{b}] \Pr[E_{b}]\\
    &= \Pr[E_{b}] \left( \frac{1}{\delta} \Pr[\fp(x^{*}) \mid E_{b}] - \frac{1}{1 - \delta}\Pr[\lnot \fp(x^{*}) \mid E_{b}] \right)\\
    &= \Pr[E_{b}] \left( \frac{1}{\delta} \frac{\Pr[E_{b} \mid \fp(x^{*})] \Pr[\fp(x^{*})]}{\Pr[E_{b}]} - \frac{1}{1 - \delta} \frac{\Pr[E_{b} \mid \lnot \fp(x^{*})] \Pr[\lnot \fp(x^{*})]}{\Pr[E_{b}]} \right)\\
    &= \frac{1}{\delta} \Pr[E_{b} \mid \fp(x^{*})] \Pr[\fp(x^{*})] - \frac{1}{1 - \delta} \Pr[E_{b} \mid \lnot \fp(x^{*})] \Pr[\lnot \fp(x^{*})]
    \end{align*}
    To evaluate the overall bound, we first derive expressions for $\Pr[\fp(x^{*}]$ and $\Pr[\lnot \fp(x^{*})]$, which will be needed in later calculations. Since $E_{s}$ and $\lnot E_{s}$ are collectively exhaustive events, we can use the law of total probability.

    \begin{align*}
        \Pr[\fp(x^{*})] &= \Pr[\fp(x^{*}) \mid E_{s}] \Pr[E_{s}] +  \Pr[\fp(x^{*}) \mid \lnot E_{s}] \Pr[\lnot E_{s}]\\
        &= 1 \cdot \psaturate + \fpr (1 - \psaturate) = \psaturate + \fpr (1 - \psaturate)
    \end{align*}
    \begin{align*}
        \Pr[\lnot \fp(x^{*})] &= \Pr[\lnot \fp(x^{*}) \mid E_{s}] \Pr[E_{s}] +  \Pr[\lnot \fp(x^{*}) \mid \lnot E_{s}] \Pr[\lnot E_{s}]\\
        &= 0 \cdot \psaturate + (1 - \fpr)(1 - \psaturate) = (1 - \fpr)(1 - \psaturate)
    \end{align*}
    Since each $x_{i}$ is chosen uniformly randomly independent of $x^{*}$, and $\adv$'s betting decision $E_{b}$ depends only on the $x_{i}$s, $E_{b}$ and $\fp(x^{*})$ are independent events. \[\Pr[E_{b} \mid \fp(x^{*})] = \Pr[E_{b} \mid \lnot \fp(x^{*})] = \Pr[E_{b}]\]
    $\adv$ bets when all $t$ uniformly randomly chosen $x_{i}$s are false positive, which happens with probability $1$ if $\modifiedbloom$ is saturated and probability $\fpr^{t}$ when $\modifiedbloom$ is unsaturated. 
    \begin{align*}
        \Pr[E_{b}] &= \Pr[E_{b} \mid E_{s}] \Pr[E_{s}] + \Pr[E_{b} \mid \lnot E_{s}] \Pr[\lnot E_{s}] = \psaturate + \fpr^{t} (1 - \psaturate)
    \end{align*}
    The overall expression for $\expect{C_{\adv}}$ is then
    \begin{align*}
        \expect{C_{\adv}} &= \frac{1}{\delta} \Pr[E_{b} \mid \fp(x^{*})] \Pr[\fp(x^{*})] - \frac{1}{1 - \delta} \Pr[E_{b} \mid \lnot \fp(x^{*})] \Pr[\lnot \fp(x^{*})]\\
        &= \frac{1}{\delta} \Pr[E_{b}] \Pr[\fp(x^{*})] - \frac{1}{1 - \delta} \Pr[E_{b}] \Pr[\lnot \fp(x^{*})]\\
        &= \Pr[E_{b}] \left( \frac{1}{\delta} \Pr[\fp(x^{*})] - \frac{1}{1 - \delta} \Pr[\lnot \fp(x^{*})] \right)\\
        &= (\psaturate + \fpr^{t} (1 - \psaturate)) \left( \frac{1}{\delta} (\psaturate + \fpr (1 - \psaturate)) - \frac{1}{1 - \delta} ((1 - \fpr)(1 - \psaturate)) \right) \\
    \end{align*}
    Since $\fpr, \psaturate \in (0, 1)$, we can set $\fpr = 0$ to get,
    \begin{align*}
        \expect{C_{\adv}} &\geq \psaturate \left( \frac{1}{\delta} \psaturate - \frac{1}{1 - \delta} (1 - \psaturate) \right) = \frac{1}{\delta} \psaturate^{2} - \frac{1}{1 - \delta} \psaturate (1 - \psaturate)
    \end{align*}
    The condition for this lower bound to be strictly positive is 
    \begin{align*}
        \frac{1}{\delta} \psaturate^{2} - \frac{1}{1 - \delta} \psaturate (1 - \psaturate) > 0 &\implies \frac{1}{\delta} \psaturate^{2} > \frac{\psaturate (1 - \psaturate)}{1 - \delta}
    \end{align*}
    Since $\psaturate > 0$, we can divide by $\psaturate$, to get $\psaturate (1 - \delta) > \delta (1 - \psaturate)$ which is true when $\psaturate > \delta$. This proves that $\modifiedbloom$ is not $(n, t, \delta)$-resilient under the BP-test for any $\delta < \psaturate$, which is the statement of the theorem.
\end{proof}

This result shows that even replacing hash functions with ideal PRFs or random functions does not prevent the BP-test attack. The attack exploits the saturation of the bit array. If every query returns $1$, the adversary can bet with non-negligible expected profit. The condition $\delta < \psaturate$ holds for a large number of non-trivial Standard Bloom filters used in practice. Since $\psaturate \geq 1 - me^{-nk/m}$, if $\delta < 1 - me^{-nk/m}$ then $\delta < \psaturate$. $\delta < 1 - me^{-nk/m}$ is equivalent to $me^{-nk/m} > 1 - \delta$. This evaluates to $nk > m\ln{\frac{1 - \delta}{m}}$. A common method to approximate (but not calculate exactly~\cite{Bose2008}) the optimal number of hash functions, $k$, in a Standard Bloom filter is $k = \frac{m}{n} \ln{2}$, as analyzed in~\cite{MitzenmacherBroder2004}. Using this expression for $k$, the bound for $\delta$ becomes $n \frac{m}{n} \ln{2} > m \ln{\frac{1 - \delta}{m}}$ which is $2m > 1 - \delta$ or more simply $\delta > 1 - 2m$. Since any non-trivial Standard Bloom filter uses at least $1$ bit, we can assume $m \geq 1$. This simplifies the bound to $\delta > -1$, which is always true since $\delta \in (0, 1)$. Thus if we apply the $k = \frac{m}{n}\ln{2}$ approximation, $\modifiedbloom$ is not $(n, t, \delta)$-resilient under the BP test for any $t \in \NN$ and any $\delta \in (0, 1)$.

%% file: sections/conclusion.tex
\section{Conclusion and Future Work}

This work advances the theory of adversarially robust Bloom filters by solving three open problems and clarifying the structure of the space. We presented the first Bloom filter constructions satisfying differential privacy guarantees, both symmetric and asymmetric, without altering query semantics. We established the first provable reduction between the simulator-based model of \filic{} et al. and the game-based model of Naor et al., showing that \filic{} correctness implies AB-test resilience. We also resolved a key open problem by proving that PRF-backed Standard Bloom filters are not resilient to the BP-test.

Our taxonomy organizes the landscape of adversarial models, robustness definitions, and privacy goals, exposing several natural but unresolved questions. In particular, we leave open whether the \filic{} model can be extended to learned Bloom filters, whether dynamic or repeated-query versions of the NOY tests can be defined, and how the Bender and CPS models relate to the more widely adopted NOY and \filic{} frameworks. We hope this work provides a foundation for developing a unified theory of privacy and robustness in probabilistic data structures.

%% file: sections/acknowledgements.tex
\section*{Acknowledgements}

We thank anonymous reviewers for helpful feedback and corrections. We thank Dr. Allison Bishop for helpful insights from her graduate course in Data Privacy at the City College of New York.

%% file: secure-cbf-tcc-2025.bbl
\newcommand{\etalchar}[1]{$^{#1}$}
\begin{thebibliography}{RSMW{\etalchar{+}}22}

\bibitem[ABT24]{almashaqbeh2024adversary}
Ghada Almashaqbeh, Allison Bishop, and Hayder Tirmazi.
\newblock Adversary resilient learned bloom filters.
\newblock {\em arXiv preprint arXiv:2409.06556}, 2024.

\bibitem[AGH70]{abernathy1970estimates}
James~R Abernathy, Bernard~G Greenberg, and Daniel~G Horvitz.
\newblock Estimates of induced abortion in urban north carolina.
\newblock {\em Demography}, 7:19--29, 1970.

\bibitem[BBL12]{bianchi2012better}
Giuseppe Bianchi, Lorenzo Bracciale, and Pierpaolo Loreti.
\newblock ” better than nothing” privacy with bloom filters: To what extent?
\newblock In {\em International Conference on Privacy in Statistical Databases}, pages 348--363. Springer, 2012.

\bibitem[BFCG{\etalchar{+}}18]{bender2018bloom}
Michael~A Bender, Martin Farach-Colton, Mayank Goswami, Rob Johnson, Samuel McCauley, and Shikha Singh.
\newblock Bloom filters, adaptivity, and the dictionary problem.
\newblock In {\em 2018 IEEE 59th Annual Symposium on Foundations of Computer Science (FOCS)}, pages 182--193. IEEE, 2018.

\bibitem[BGK{\etalchar{+}}08]{Bose2008}
Prosenjit Bose, Hua Guo, Evangelos Kranakis, Anil Maheshwari, Pat Morin, Jason Morrison, Michiel Smid, and Yihui Tang.
\newblock On the false-positive rate of bloom filters.
\newblock {\em Information Processing Letters}, 108(4):210--213, 2008.

\bibitem[Blo70]{bloom1970}
Burton~H Bloom.
\newblock Space/time trade-offs in hash coding with allowable errors.
\newblock {\em Communications of the ACM}, 13(7):422--426, 1970.

\bibitem[BM04]{MitzenmacherBroder2004}
Andrei Broder and Michael Mitzenmacher.
\newblock Network applications of bloom filters: A survey.
\newblock {\em Internet mathematics}, 1(4):485--509, 2004.

\bibitem[CPS19]{claytonetal}
David Clayton, Christopher Patton, and Thomas Shrimpton.
\newblock Probabilistic data structures in adversarial environments.
\newblock In {\em ACM SIGSAC Conference on Computer and Communications Security (CCS)}, 2019.

\bibitem[DR{\etalchar{+}}14]{dwork2014algorithmic}
Cynthia Dwork, Aaron Roth, et~al.
\newblock The algorithmic foundations of differential privacy.
\newblock {\em Foundations and Trends{\textregistered} in Theoretical Computer Science}, 9(3--4):211--407, 2014.

\bibitem[FKKU25]{filic2025deletions}
Mia Fili{\'c}, Keran Kocher, Ella Kummer, and Anupama Unnikrishnan.
\newblock Deletions and dishonesty: Probabilistic data structures in adversarial settings.
\newblock In {\em International Conference on the Theory and Application of Cryptology and Information Security}, pages 137--168. Springer, 2025.

\bibitem[Fou23]{kernelbpfmap}
Linux Foundation.
\newblock {Linux Kernel Documentation - BPF Maps}.
\newblock \url{https://www.kernel.org/doc/html/next/bpf/maps.html}, 2023.
\newblock Accessed: 2023-05-02.

\bibitem[FPUV22]{filic1}
Mia Filic, Kenneth~G Paterson, Anupama Unnikrishnan, and Fernando Virdia.
\newblock Adversarial correctness and privacy for probabilistic data structures.
\newblock In {\em ACM SIGSAC Conference on Computer and Communications Security (CCS)}, 2022.

\bibitem[GKL14]{gerbetsecurity}
Thomas Gerbet, Amrit Kumar, and C{\'e}dric Lauradoux.
\newblock On the (in) security of google safe browsing.
\newblock {\em INRIA ePrint}, 2014.

\bibitem[GKL15]{evilchoicesbloom}
Thomas Gerbet, Amrit Kumar, and C{\'e}dric Lauradoux.
\newblock The power of evil choices in bloom filters.
\newblock In {\em IEEE/IFIP International Conference on dependable systems and networks}, 2015.

\bibitem[Goo23]{leveldbbloom}
Google.
\newblock {LevelDB Bloom Filter}.
\newblock \url{https://github.com/google/leveldb/blob/main/util/bloom.cc}, 2023.
\newblock Accessed: 2023-05-04.

\bibitem[GRW{\etalchar{+}}22]{galan2022privacy}
Sergio Gal{\'a}n, Pedro Reviriego, Stefan Walzer, Alfonso S{\'a}nchez-Macian, Shanshan Liu, and Fabrizio Lombardi.
\newblock On the privacy of counting bloom filters under a black-box attacker.
\newblock {\em IEEE Transactions on Dependable and Secure Computing}, 20(5):4434--4440, 2022.

\bibitem[HLM08]{hox_mulders}
Joop Hox and Gerty Lensvelt-Mulders.
\newblock Encyclopedia of survey research methods, 2008.
\newblock Randomized Response.

\bibitem[KLS{\etalchar{+}}25]{ke2025dpbloomfilter}
Yekun Ke, Yingyu Liang, Zhizhou Sha, Zhenmei Shi, and Zhao Song.
\newblock Dpbloomfilter: Securing bloom filters with differential privacy.
\newblock {\em arXiv preprint arXiv:2502.00693}, 2025.

\bibitem[LN25]{lotan2025adversarially}
Chen Lotan and Moni Naor.
\newblock Adversarially robust bloom filters: Monotonicity and betting.
\newblock {\em IACR Communications in Cryptology}, 2(1), 2025.

\bibitem[Man94]{mangat1994improved}
Naurang~S Mangat.
\newblock An improved randomized response strategy.
\newblock {\em Journal of the Royal Statistical Society: Series B (Methodological)}, 56(1):93--95, 1994.

\bibitem[Met]{rocksdbbloom}
Meta.
\newblock {RocksDB Bloom Filter}.
\newblock \url{https://github.com/facebook/rocksdb/blob/main/util/dynamic_bloom.h}.
\newblock Accessed: 2023-05-04.

\bibitem[MPR20]{mitzenmacher2020adaptive}
Michael Mitzenmacher, Salvatore Pontarelli, and Pedro Reviriego.
\newblock Adaptive cuckoo filters, 2020.

\bibitem[MU17]{mitzenmacher2017probability}
Michael Mitzenmacher and Eli Upfal.
\newblock {\em Probability and computing: Randomization and probabilistic techniques in algorithms and data analysis}.
\newblock Cambridge university press, 2017.

\bibitem[NE19]{moni1}
Moni Naor and Yogev Eylon.
\newblock Bloom filters in adversarial environments.
\newblock {\em ACM Transactions on Algorithms (TALG)}, 15(3):1--30, 2019.

\bibitem[NO22]{naor2022bet}
Moni Naor and Noa Oved.
\newblock Bet-or-pass: Adversarially robust bloom filters.
\newblock In {\em Theory of Cryptography Conference}, 2022.

\bibitem[RHDS21]{reviriego1}
Pedro Reviriego, Jos{\'e}~Alberto Hern{\'a}ndez, Zhenwei Dai, and Anshumali Shrivastava.
\newblock Learned bloom filters in adversarial environments: A malicious url detection use-case.
\newblock In {\em IEEE International Conference on High Performance Switching and Routing (HPSR)}. IEEE, 2021.

\bibitem[RSMW{\etalchar{+}}22]{reviriego2022privacy}
Pedro Reviriego, Alfonso S{\'a}nchez-Macian, Stefan Walzer, Elena Merino-G{\'o}mez, Shanshan Liu, and Fabrizio Lombardi.
\newblock On the privacy of counting bloom filters.
\newblock {\em IEEE Transactions on Dependable and Secure Computing}, 20(2):1488--1499, 2022.

\bibitem[SBBR17]{sengupta2017sampling}
Neha Sengupta, Amitabha Bagchi, Srikanta Bedathur, and Maya Ramanath.
\newblock Sampling and reconstruction using bloom filters.
\newblock {\em IEEE Transactions on Knowledge and Data Engineering}, 30(7):1324--1337, 2017.

\bibitem[Swe02]{sweeney2002k}
Latanya Sweeney.
\newblock k-anonymity: A model for protecting privacy.
\newblock {\em International journal of uncertainty, fuzziness and knowledge-based systems}, 10(05):557--570, 2002.

\bibitem[TKCY22]{takagi2022asymmetric}
Shun Takagi, Fumiyuki Kato, Yang Cao, and Masatoshi Yoshikawa.
\newblock Asymmetric differential privacy.
\newblock In {\em 2022 IEEE International Conference on Big Data (Big Data)}, pages 1576--1581. IEEE, 2022.

\bibitem[VF24]{virdia2024note}
Fernando Virdia and Mia Fili{\'c}.
\newblock A note on securing insertion-only cuckoo filters.
\newblock {\em Cryptology ePrint Archive}, 2024.

\bibitem[VKMK21]{plbf}
Kapil Vaidya, Eric Knorr, Michael Mitzenmacher, and Tim Kraska.
\newblock Partitioned learned bloom filters.
\newblock In {\em International Conference on Learning Representations}, 2021.

\bibitem[War65]{warner1965randomized}
Stanley~L Warner.
\newblock Randomized response: A survey technique for eliminating evasive answer bias.
\newblock {\em Journal of the American statistical association}, 60(309):63--69, 1965.

\bibitem[XQZ{\etalchar{+}}14]{xiangyu2014randomized}
Chen Xiangyu, Du~Qiaoqiao, Jin Zongda, Xu~Tian, Shi Jiachen, and Gao Ge.
\newblock The randomized response technique application in the survey of homosexual commercial sex among men in beijing.
\newblock {\em Iranian journal of public health}, 43(4):416, 2014.

\end{thebibliography}
